\pdfoutput=1
\documentclass[a4paper]{article}
\usepackage[utf8]{inputenc}
\usepackage[T1]{fontenc}

\usepackage[
  colorlinks=true,
  linkcolor=darkgray,
  citecolor=darkgray,
  urlcolor=darkgray
]{hyperref}
\usepackage{amsmath,amsthm,amssymb}
\usepackage{authblk}
\usepackage{circuitikz}
\usepackage{doi}
\usepackage{enumitem}
\usepackage{lmodern}
\usepackage{microtype}
\usepackage[square, numbers]{natbib}
\usepackage{orcidlink}
\usepackage{tabularx}
\usepackage{tikz}
\usepackage{xcolor}
\usepackage{xspace}

\newcommand{\cc}{\mathbf{c}}
\newcommand{\e}{\mathbf{e}}

\newcommand{\z}{\mathbf{z}}
\newcommand{\vv}{\mathbf{v}}
\newcommand{\ww}{\mathbf{w}}

\newcommand{\CC}{\mathbb{C}}
\newcommand{\QQ}{\mathbb{Q}}
\newcommand{\QQGE}{\QQ_{\geqslant 0}}
\newcommand{\NN}{\mathbb{N}}
\newcommand{\RR}{\mathbb{R}}
\newcommand{\ZZ}{\mathbb{Z}}

\newcommand{\AAA}{\mathcal{A}}

\newcommand{\FFF}{\mathcal{F}}

\newcommand{\LLL}{\mathcal{L}}
\newcommand{\NNN}{\mathcal{N}}
\newcommand{\NNNC}{\NNN^-}
\newcommand{\NNNR}{\NNN^+}
\newcommand{\NNNRR}{\NNN^{(\RR)}}
\newcommand{\QQQ}{\mathcal{Q}}

\newcommand{\TTT}{\mathcal{T}}
\newcommand{\VVV}{\mathcal{V}}

\newcommand{\ACF}{\textrm{ACF}_0}
\newcommand{\RCF}{\textrm{RCF}}

\newcommand{\cj}[1]{\overline{#1}}
\newcommand{\ii}{\mathord{\mathrm{i}}}
\newcommand{\im}{\Im}
\newcommand{\re}{\Re}

\newcommand{\mult}{\mathbin{{\, \cdot \, }}}

\newcommand{\sq}[1]{\lq #1\rq}

\newcommand{\tequiv}{\approx}

\newcommand{\myland}{\mathrel{\land}}
\newcommand{\mylor}{\mathrel{\lor}}
\newcommand{\mylongrightarrow}{\mathrel{\longrightarrow}}
\newcommand{\mylongleftrightarrow}{\mathrel{\longleftrightarrow}}

\newcommand{\mcpec}{\textsc{mcpec}\xspace}
\newcommand{\cad}{\textsc{cad}\xspace}
\newcommand{\qe}{\textsc{qe}\xspace}

\newtheorem{theorem}{Theorem}
\newtheorem{corollary}{Corollary}
\newtheorem{proposition}{Proposition}
\theoremstyle{definition}
\newtheorem{definition}{Definition}
\newtheorem{lemma}{Lemma}
\theoremstyle{remark}
\newtheorem{example}{Example}
\newtheorem{remark}{Remark}

\begin{document}

\title{Pseudo-Complex Quantifier Elimination}

\author[1]{Nicolas Faroß\,\orcidlink{0009-0001-4419-2337}\,}
\author[2,3]{Thomas Sturm\,\orcidlink{0000-0002-8088-340X}\,}
\affil[1]{Chalmers University of Technology and\authorcr University of Gothenburg, Gothenburg, Sweden\authorcr
\href{mailto:faross@chalmers.se}{faross@chalmers.se}}

\affil[2]{University of Lorraine, CNRS, Inria, LORIA, Nancy, France\authorcr
\href{mailto:thomas.sturm@cnrs.fr}{thomas.sturm@cnrs.fr}}

\affil[3]{MPI-INF and Saarland University, SIC, Saarbrücken, Germany\authorcr
\href{mailto:sturm@mpi-inf.mpg.de}{sturm@mpi-inf.mpg.de}}

\maketitle

\begin{abstract}
We describe the design of a quantifier elimination framework for the complex numbers in the language of ordered rings supplemented with symbols for the imaginary unit, real parts, imaginary parts, and conjugates. Technically, we use a reduction to real quantifier elimination followed by a heuristic reinterpretation of the results within our complex framework. We present computational examples using a prototypical implementation of our approach in our Python-based open-source system Logic1.
\end{abstract}

\section{Introduction}
Based on research conducted during the 1930s \cite{Tarski1930}, Tarski published the first complete quantifier elimination (\qe) procedure for the real numbers in 1948 \cite{Tarski:48a}. Subsequent work by Seidenberg in 1954 provided a geometric interpretation of these results, known today as the Tarski--Seidenberg Theorem \cite{Seidenberg1954}, which naturally gives rise to the idea of applying these methods to real-world problems. Remarkably, concluding remarks in a 1954 report to the US Army by Davis already address the infeasibility of implementing real \qe, indicating early interest in the problem \cite{Davis:54a}. Given the state of the available programming infrastructure at the time, along with the fact that the Tarski--Seidenberg procedure is not even elementary recursive in the worst case, Davis's assessment was certainly justified. This picture began to change during the 1970s, when Collins and his students developed and implemented real \qe methods based on cylindrical algebraic decomposition (\cad) \cite{Collins:75}. In the mid-1980s partial \cad constituted a major breakthrough \cite{CollinsHong:91}. Complemented with more specialized virtual substitution methods aimed at formulas of low degree \cite{Weispfenning:88a,Weispfenning:97b,Kosta:16a}, real \qe became applicable to a considerable number of application areas with a significant publication record, also at the CASC conference series. Another class of implemented approaches to real \qe is based on multivariate parametric real root counting \cite{Weispfenning:98a,DolzmannSturm:99a,LeSafeyElDin2021}.

The early history of \qe for the ring of complex numbers resembles that of real \qe outlined above. It was again Tarski who developed a first \qe procedure during the 1940s, without explicitly publishing it at the time \cite {Tarski1954}. With his famous theorem on the constructability of images, Chevalley contributed a geometric perspective on Tarski's result for the complex numbers, resembling Seidenberg's work for the reals \cite{Chevalley1943}. As in the real case, Tarski's procedure is not elementary recursive. A more modern approach, implemented in Redlog \cite{DolzmannSturm:97a}, is based on comprehensive Gröbner bases \cite{Weispfenning:92a}. More generally, complex \qe is supported in computer algebra systems through algebraic decomposition techniques. In particular, Maple’s Regular Chains library \cite{ChenMorenoMaza2012} provides functionality for manipulating constructible sets and computing projections via triangular decomposition \cite{AubryLazardMorenoMaza1999}. Other systems such as Singular, Magma, and Macaulay2 offer the necessary elimination primitives, but require manual construction of \qe procedures.

When focusing on rigorous logical settings, real \qe is currently much better supported in software than complex \qe, and numerous applications have been documented in the scientific literature; see, e.g., \cite{DSW:98,DolzmannSturm:99a,Sturm:17a} and the references therein. One recent line of research focuses on the qualitative analysis of chemical reaction networks, e.g., \cite{RahkooySturm:21a}. A major reason for the stronger interest in real \qe is that classical complex \qe in the language of rings lacks symbols naturally expected by users in the natural sciences and engineering, namely a constant for the imaginary unit, as well as operators for the real and imaginary parts and for complex conjugation. In fact, the addition of these symbols implicitly reduces the problem to real \qe, as is essentially known in the community. Our work makes this approach explicit by providing a formal \qe procedure in the extended language, which we call \emph{pseudo-complex quantifier elimination}.

Our original contributions are the following:

\begin{enumerate}
    \item We formally introduce terms and first-order formulas, define various relevant syntactic normal forms, and provide a rigorous formal semantics. (Sections~\ref{SE:terms} and~\ref{SE:formulas})
    \item We reduce \qe for our framework to real \qe. We furthermore prove that the first-order theory corresponding to our framework is complete and decidable. (Section~\ref{SE:qe})
    \item Building on real \qe, we obtain \qe results with terms inductively constructed from rational constants and the real and imaginary parts $\re(z)$ and $\im(z)$ of complex variables $z$, exclusively using ring operations. We introduce graph-based heuristics to produce equivalent, more compact results in the unrestricted language and prove complexity bounds for these heuristics. (Section~\ref{SE:reinter})
    \item We give asymptotic complexity bounds for our \qe, which coincide with the known bounds for classical real and complex \qe. From a practical perspective, our approach makes existing real \qe software accessible for problems over complex numbers. (Section~\ref{SE:complexity})
    \item We have implemented our approach in our Python-based open-source system Logic1. A number of computation examples demonstrate the practical applicability of our approach. (Section~\ref{SE:implementation})
\end{enumerate}
We finally summarize and evaluate our results in Section~\ref{SE:conclusions}.

\section{Language, Terms and Normal Forms}\label{SE:terms}

We use the countably infinite \emph{language} of ordered rings along with constant symbols for all non-negative rational numbers, plus another constant symbol for the imaginary unit, plus further unary function symbols for the real part, the imaginary part, and the conjugate:
\begin{equation*}
    \LLL = \bigl\{\,
        q^{(0)},\,
        +^{(2)},\,
        -^{(1)},\,
        \mult^{(2)},\,
        \ii^{(0)},\,
        \re^{(1)},\,
        \im^{(1)},\,
        \cj{\,\cdot\, }^{(1)};\,
        <^{(2)},\,
        \leq^{(2)},\,
        \neq^{(2)}
    \,\bigm\vert\,
        q \in \QQGE
    \, \bigr\}.
\end{equation*}
Note that \sq{$<$}, \sq{$\leq$}, and \sq{${\neq}$} are relation symbols, while all other symbols in $\LLL$ are function symbols, including constant symbols. Numbers in parentheses indicate the arities of respective symbols. We furthermore fix a countably infinite, strictly ordered set $(\VVV, \prec)$ of \emph{variables}. As usual, \emph{terms} are built inductively, starting with constant symbols and variables and applying non-constant function symbols, respecting their arities. This yields the countably infinite set $\TTT$ of all terms with function symbols from $\LLL$ and variables from $\VVV$.

Let $n \in \NN$ and $z_1$, \dots,~$z_n \in \VVV$ with $z_1 \prec \dots \prec z_n$, and set $\z = (z_1, \dots, z_n)$. We write $\TTT_\z \subseteq \TTT$ for the set of all terms over variables in $\z$.

Consider $t \in \TTT$ and let $\{z_1, \dots, z_n\} \subseteq \VVV$ be a superset of the variables occurring in $t$. Then $t(z_1, \dots, z_n)$ is called an \emph{extended term}. Each extended term induces an \emph{interpretation} $t(z_1, \dots, z_n)^*: \CC^n \to \CC$. Two terms $t_1$, $t_2 \in \TTT_\z$ are \emph{equivalent} if $t_1(\z)^* = t_2(\z)^*$. We then write $t_1 \tequiv t_2$.

\begin{remark}[Universal statements about extended terms]\label{REM:universal}
    For $t_1$, $t_2 \in \TTT_\z$ the following are equivalent:
    \begin{enumerate}[label=(\roman*), leftmargin=*]
        \item $t_1 \tequiv t_2$;
        \item $t_1(\z')^* = t_2(\z')^*$ for at least one $m \in \NN$ and $\z' \in \VVV^m$ such that both $t_1(\z')$ and $t_2(\z')$ are extended terms;
        \item $t_1(\z')^* = t_2(\z')^*$ for all $m \in \NN$ and $\z' \in \VVV^m$ such that both $t_1(\z')$ and $t_2(\z')$ are extended terms.
    \end{enumerate}
    The reason is that modification of the extension by permutation, addition, or removal of variables not occurring in $t_1$, $t_2$ does not affect the universal statement about all possible complex arguments in the definition of term equivalence.\qed
\end{remark}

We call $\NNN_\z \subseteq \TTT_\z$ a system of \emph{normal forms} for $\TTT_\z$ if for each $t \in \TTT_\z$ there is $t' \in \NNN_\z$ with $t' \tequiv t$. We call $\NNN_\z$ a system of \emph{unique} normal forms if there is exactly one such $t'$ for each $t$.

Let $t_1$, \dots, $t_n \in \TTT$. As usual, we agree that multiplication has higher precedence than addition. Furthermore, we shortly write $t_1 + \dots + t_n$ and $t_1 \cdots t_n$ for iterated left-associative application of addition and multiplication, respectively. Finally, exponentiation $t_1^m$ denotes $m$-fold multiplication for $m \geq 2$.

Denote by $\QQ[\ii]$ the algebraic extension ring, whose elements can be uniquely written as linear polynomials $q_1\ii + q_2$ with $q_1$, $q_2 \in \QQ$. We note that $\QQ[\ii]$ is even a field, which equals the extension field $\QQ(\ii)$. On these grounds, we specify \emph{conjugate normal forms} $\NNNC_\z$ and \emph{Cartesian normal forms} $\NNNR_\z$ for $\TTT_\z$  as follows.

\begin{proposition}[Conjugate normal form]\label{PR:cnf}
    Consider $\TTT_\z$ with $\z \in \VVV^n$ for $n \in \NN$. Define the polynomial ring $R = \QQ[\ii][z_1, \dots, z_n, \cj{z_1}, \dots, \cj{z_n}]$ with $2n$ algebraically independent generators. Let $\NNNC_\z \subseteq \TTT_\z$ be the set of terms corresponding to polynomials in $R$ using linear polynomial coefficients and degree-lexicographic order. Then $\NNNC_\z$ is a system of unique normal forms for $\TTT_\z$.
\end{proposition}

\begin{proof}
    Let $t \in \TTT_\z$. We transform $t$ in three subsequent steps as follows:
\begin{enumerate}
    \item Equivalently rewrite subterms of $t$ as follows, where $s_1$ does not contain $\re$ and $s_2$ does not contain $\im$:
    \begin{equation*}
        \re(s_1) \to \frac{1}{2} \cdot (s_1 + \cj{s_1}),
        \quad
        \im(s_2) \to -\frac{1}{2} \cdot \ii \cdot (s_2 - \cj{s_2}).
    \end{equation*}
    This procedure terminates and yields $t' \in \TTT_\z$ with $t' \tequiv t$, and $t'$ does not contain $\re$, $\im$.
    \item In $t'$, equivalently propagate complex conjugates inwards by rewriting subterms as follows:
    \begin{align*}
        &\cj{s_1 + s_2} \to \cj{s_1} + \cj{s_2},\quad
        \cj{s_1 \cdot s_2} \to \cj{s_1} \cdot \cj{s_2},\quad
        & s_1, s_2 &\in \TTT_\z,\\
        &\cj{-s} \to -\cj{s},\quad
        \cj{\cj{s}} \to s,
        & s &\in \TTT_\z,\\
        &\cj{q} \to q,
        & q &\in \QQGE,\\
        &\cj{\ii} \to -\ii.
    \end{align*}
    This procedure terminates and yields $t'' \in \TTT_\z$ with $t'' \tequiv t$, and $t''$ is exclusively built from variables, constant symbols, conjugates of variables, and ring operator symbols.
    \item We equivalently expand $t''$ into a polynomial $t''' \in R$ with $t''' \tequiv t$, using the laws of arithmetic of commutative rings.
\end{enumerate}
    Uniqueness of $t'''$ follows from the observation that the set
    \begin{equation*}
        \bigl\{\,z_1^{e_1} \cdots z_n^{e_n}\cdot \cj{z_1}^{e_{n+1}} \cdots \cj{z_n}^{e_{2n}}(\z)^* \bigm| \e \in \NN^{2n}\,\bigr\}
    \end{equation*}
    is linear independent in the vector space of functions from $\CC^n$ to $\CC$, so that our normal forms are identical zero as functions if and only if they are formally equal to zero.
\end{proof}

Note that we can obtain alternative unique conjugate normal forms for $\TTT_\z$ by using other monomial orders instead of the degree-lexicographic order.

\begin{proposition}[Cartesian normal norm]\label{PR:cartnf}
    Consider $\TTT_\z$ with $\z \in \VVV^n$ for $n \in \NN$. Define the polynomial ring $S = \QQ[\re(z_1), \dots, \re(z_n), \im(z_1), \dots, \im(z_n)]$ with $2n$ algebraically independent generators. Let $\NNNR_\z = \{\, p_1 + \ii \cdot p_2 \in \TTT_\z \mid p_1, p_2 \in S\,\}$ using degree-lexicographic order for $p_1$ and $p_2$ and omitting summands $0$ and factors $1$ as usual. Then $\NNNR_\z$ is a system of unique normal forms for $\TTT_\z$.
\end{proposition}

\begin{proof}
    Let $t \in \TTT_\z$, without loss of generality $t \in \NNNC_\z$. Recall from Proposition~\ref{PR:cnf} that $t$ corresponds to a polynomial in $\QQ[\ii][\z, \cj{\z}]$, where the coefficients are represented as linear polynomials in $\ii$. We proceed in three steps:
    \begin{enumerate}
        \item Denote $\z = (z_1, \dots, z_n)$ and rewrite the monomials of $t$ as follows:
        \begin{equation*}
            \prod_{i=1}^n z_i^{e_i} \cdot \prod_{i=1}^n \cj{z_i}^{e_{n + i}}\to
            \prod_{i=1}^n \bigl(\re(z_i) + \ii\cdot\im(z_i)\bigr)^{e_i} \cdot
            \prod_{i=1}^n \bigl(\re(z_i) - \ii\cdot\im(z_i)\bigr)^{e_{n + i}}.
        \end{equation*}
        This yields $t' \in \TTT_\z$ with $t' \tequiv t$, and $t'$ is exclusively built from constant symbols, subterms $\re(z_i)$, $\im(z_i)$, and ring operator symbols.
    \item Equivalently expand $t'$ into a polynomial
    \begin{equation*}
        t'' \in \QQ[\ii][\re(z_1), \dots, \re(z_n), \im(z_1), \dots, \im(z_n)]
    \end{equation*}
    with $t'' \tequiv t$, using the laws of arithmetic of commutative rings.
    \item Using the laws of arithmetic of commutative rings once again, it is straightforward to equivalently transform $t''$ into $t''' \in \NNNR_\z$.
    \end{enumerate}
    For uniqueness, it is sufficient to show that our normal forms $t'''$ are pairwise not equivalent. Consider $t_1'''$, $t_2'''$ with $t_1''' \neq t_2'''$. According to our construction, there are $t_1$, $t_2$ in conjugate normal form with $t_1 \neq t_2$ and $t_1''' \tequiv t_1$ and $t_2''' \tequiv t_2$. By the uniqueness of the conjugate normal form, we know that $t_1 \not\tequiv t_2$, hence $t_1''' \tequiv t_1 \not\tequiv t_2 \tequiv t_2'''$.
\end{proof}

Similarly to $\QQ[\ii][\z, \cj{\z}]$ above, we will shortly write $\QQ[\re(\z), \im(\z)]$ from now on. Again, we can obtain alternative unique Cartesian normal forms for $\TTT_\z$ by using other monomial orders instead of the degree-lexicographic order. We furthermore note that Cartesian normal forms of arbitrary terms can be more efficiently computed directly instead of going via conjugate normal forms. In the symbols $\NNNC_\z$ and $\NNNR_\z$, the \sq{$-$} and the \sq{$+$} point at the conjugate operation and the Cartesian coordinate system, respectively.

\begin{example}[Conjugate and cartesian normal forms]
    Consider
    \begin{align*}
        s_1 & = z^2+ \ii, &
        t_1 & = \bigl(\re(z)^2 - \im(z)^2\bigr) + \ii \cdot \bigl(2\re(z)\im(z) + 1\bigr),\\
        s_2 & = z\cj{z}, &
        t_2 & = \re(z)^2 + \im(z)^2,\\
        s_3 & = 2\ii + 1, &
        t_3 &= 1 + \ii \cdot 2.
    \end{align*}
    For $i \in \{1, \dots, 3\}$, we have $s_i \tequiv t_i$, where $s_i$ is in conjugate normal form, and $t_i$ is in Cartesian normal form.
\end{example}

We call $t \in \TTT_\z$ a \emph{real term} if $t(\z)^*[\CC^n] \subseteq \RR$. This is equivalent to the requirement that $t(\z')^*[\CC^m] \subseteq \RR$ for all extended terms $t(\z')$ with $\z' \in \VVV^m$; compare Remark~\ref{REM:universal}.

\begin{proposition}[Characterization of real terms]\label{PR:real_terms}
    Let $t \in \TTT_\z$ with $\z \in \VVV^n$ for $n \in \NN$. Then the following are equivalent:
    \begin{enumerate}[label=(\roman*), leftmargin=*]
        \item $t$ is a real term;
        \item $\im(t) \tequiv 0$;
        \item the Cartesian normal form of $t$ is a polynomial in $\QQ[\re(\z), \im(\z)]$.
    \end{enumerate}
\end{proposition}

\begin{proof}
    We denote by $\re^*: \CC \to \CC$ the real part, by $\im^*: \CC \to \CC$ the imaginary part, and by $\ii^* \in \CC$ the imaginary unit.

    Assume (iii), that is $t \tequiv t'$ with $t' \in \QQ[\re(\z), \im(\z)]$, and let $\cc \in \CC^n$. Since $\re^*[\CC] = \RR$, $\im^*[\CC] = \RR$, and $\RR \subseteq \CC$ is closed under ring arithmetic, we obtain $t(\z)^*(\cc) = t'(\z)^*(\cc) \in \RR$. We have thus shown (i).

    Assume (i), that is $t(\z)^*(\cc) \in \RR$, and let $\cc \in \CC^n$. It follows that
    \begin{equation*}
        \im(t)(\z)^*(\cc) = \im^*(t(\z)^*(\cc)) = 0 = 0(\z)^*(\cc),
    \end{equation*}
    and we have shown (ii).

    Assume (ii), and let $f$, $g \in \QQ[\re(\z), \im(\z)]$ be the unique choices such that $t \tequiv f + \ii \cdot g$. Using the same arguments as in our proof step from (iii) to (i) above, we conclude that $f(\z)^*(\cc) \in \RR$ and $g(\z)^*(\cc) \in \RR$ for all $\cc \in \CC^n$. Assume for a contradiction that $g \neq 0$. Then there exists $\cc \in \CC^n$ such that $g(\z)^*(\cc) \neq 0$, and we obtain
    \begin{equation*}
        \im(t)(\z)^*(\cc) = \im^*(t(\z)^*(\cc)) = \im^*(f(\z)^*(\cc) + \ii^*g(\z)^*(\cc)) = g(\z)^*(\cc) \neq 0(\z)^*(\cc),
    \end{equation*}
    which contradicts (ii). Hence $g = 0$, and we have shown (iii).
\end{proof}

Note that property (ii) in the previous lemma can be checked via normal form computation on $t$, using any system of unique normal forms of terms.

\section{Atoms and First-order Formulas}\label{SE:formulas}

Given $t_1$, $t_2$, $r_1$, $r_2 \in \TTT$, where $r_1$, $r_2$ are real terms, \emph{atomic formulas}, or \emph{atoms} for short, are of one of the following forms:
\begin{equation*}
    t_1 = t_2, \quad t_1 \neq t_2, \quad r_1 \leq r_2, \quad r_1 < r_2,
\end{equation*}
where, as usual in interpreted first order logic, the equal sign is not an element of the language $\LLL$. This yields the countably infinite set $\AAA$ of all atoms. Note that we permit \sq{$<$} and \sq{$\leq$} to be used only with real terms. Our software will recognize other uses via Proposition~\ref{PR:real_terms} and raise an exception.

As usual, \emph{quantifier-free formulas} are built inductively, starting with atoms and the constant formulas $\top^{(0)}$ (\sq{true}), $\bot^{(0)}$ (\sq{false}), and applying non-constant logical operators
\begin{equation*}
    \lnot^{(1)},\quad \myland^{(*)},\quad \mylor^{(*)},\quad \mylongrightarrow^{(2)},\quad \mylongleftrightarrow^{(2)},
\end{equation*}
where $\myland$ and $\mylor$ have arbitrary arity. This yields the countably infinite set $\QQQ$ all quantifier-free formulas. We have countably infinitely many \emph{existential} and \emph{universal quantifiers}
\begin{equation*}
    \exists z^{(1)},\quad \forall z^{(1)}\quad \text{for}\quad z \in \VVV.
\end{equation*}
Prenex \emph{first-order formulas} are built inductively, starting with quantifier-free formulas and applying existential and universal quantifiers. This yields the countably infinite set $\FFF$ of all prenex first-order formulas.

We say that a variable $z \in \VVV$ \emph{occurs} in a formula $\varphi \in \FFF$ if it appears in an atom of $\varphi$. If $z$ additionally occurs in a quantifier in $\varphi$, we say that $z$ is a \emph{bound variable} of $\varphi$, otherwise $z$ is a \emph{free variable} of $\varphi$. Let $n \in \NN$ and $z_1$, \dots,~$z_n \in \VVV$ with $z_1 \prec \dots \prec z_n$, and set $\z = (z_1, \dots, z_n)$. We write $\AAA_\z \subseteq \AAA$ and $\QQQ_\z \subseteq \QQQ$ for the set of all atoms and quantifier-free formulas over variables in $\z$, respectively. We write $\FFF_\z \subseteq \FFF$ for the set of all those first-order formulas $\varphi$ for which $\z$ covers at least the free variables of $\varphi$. A prenex \emph{sentence} is a prenex first-order formula $\vartheta \in \FFF$ without free variables; $\FFF_{()}$ is the set of all sentences.

Consider $\varphi \in \FFF$, and let $\{z_1, \dots, z_n\} \subseteq \VVV$ be a superset of the free variables of $\varphi$. Then $\varphi(z_1, \dots, z_n)$ is called an \emph{extended formula}. Each extended formula $\varphi(\z)$ induces an \emph{interpretation} $\varphi(\z)^*: \CC^n \to \{0, 1\}$, which is the characteristic function of a relation, inductively defined for $\cc \in \CC^n$ as follows:
\begin{enumerate}
    \item $(t_1 = t_2)(\z)^*(\cc) = 1$ iff $t_1(\z)^*(\cc) = t_2(\z)^*(\cc)$,
    \item $(t_1 \neq t_2)(\z)^*(\cc) = 1$ iff $t_1(\z)^*(\cc) \neq t_2(\z)^*(\cc)$,
    \item $(t_1 < t_2)(\z)^*(\cc) = 1$ iff $t_1(\z)^*(\cc) < t_2(\z)^*(\cc)$, where $t_1(\z)^*(\cc)$, ${t_2(\z)^*(\cc) \in \RR}$,
    \item $(t_1 \leq t_2)(\z)^*(\cc) = 1$ iff $t_1(\z)^*(\cc) \leq t_2(\z)^*(\cc)$, where $t_1(\z)^*(\cc)$, ${t_2(\z)^*(\cc) \in \RR}$,
    \item $\top(\z)^*(\cc) = 1$,
    \item $\bot(\z)^*(\cc) = 0$,
    \item $(\lnot\varphi)(\z)^*(\cc) = 1$ iff $\varphi(\z)^*(\cc) = 0$,
    \item $(\varphi_1 \myland \dots \myland \varphi_m)(\z)^*(\cc) = \min \{\varphi_1(\z)^*(\cc), \dots, \varphi_m(\z)^*(\cc)\}$,
    \item $(\varphi_1 \mylor \dots \mylor \varphi_m)(\z)^*(\cc) = \max \{\varphi_1(\z)^*(\cc), \dots, \varphi_m(\z)^*(\cc)\}$,
    \item $(\varphi_1 \mylongrightarrow \varphi_2)(\z)^*(\cc) = 1$ iff $\varphi_1(\z)^*(\cc) \leq \varphi_2(\z)^*(\cc)$,
    \item $(\varphi_1 \mylongleftrightarrow \varphi_2)(\z)^*(\cc) = 1$ iff $\varphi_1(\z)^*(\cc) = \varphi_2(\z)^*(\cc)$,
    \item $\exists z(\varphi)(\z)^*(\cc) = \max \{\,\varphi(\z, z)^*(\cc, c) \mid c \in \CC\,\}$, where w.l.o.g.~$z \notin \{z_1, \dots, z_n\}$,\label{PAGE:semantics}
    \item $\forall z(\varphi)(\z)^*(\cc) = \min \{\,\varphi(\z, z)^*(\cc, c) \mid c \in \CC\,\}$, where w.l.o.g.~$z \notin \{z_1, \dots, z_n\}$.
\end{enumerate}
If $\varphi(\z)^*(\cc) = 1$ for specific $\cc \in \CC^n$, we write $\CC \models \varphi(\cc)$, where the extension $\z$ will be clear from the context. If $\varphi(\z)^*(\cc) = 1$ for all $\cc \in \CC^n$, we call $\varphi$ a \emph{valid formula} and write $\CC \models \varphi$, which does not depend on the extension $\z$; compare Remark~\ref{REM:universal}. According to our definition in the previous section, two terms are equivalent, $t_1 \tequiv t_2$, if and only if $\CC \models t_1 = t_2$. We say that two formulas $\varphi_1$ and $\varphi_2$ are \emph{equivalent} if $\CC \models \varphi_1 \mylongleftrightarrow \varphi_2$.

We call $\NNN_\z \subseteq \FFF_\z$ a system of \emph{normal forms} for $\FFF_\z$ if for each $\varphi \in \FFF_\z$ there is $\varphi' \in \NNN_\z$ with $\CC \models \varphi' \mylongleftrightarrow \varphi$; we denote $\varphi'$ by $\NNN_\z(\varphi)$. A first-order formula $\varphi \in \FFF_\z$ is called a \emph{real formula} if both sides of all its atoms are real terms.

\begin{proposition}[Real normal form of formulas]
    Consider $\FFF_\z$ with $\z \in \VVV^n$ for $n \in \NN$. Let $\NNNRR_\z$ be the set of all
    formulas $\varphi \in \FFF_\z$ with the following properties:
    \begin{enumerate}[label=(\roman*), leftmargin=*]
        \item all left-hand sides of atoms in $\varphi$ are polynomials in $\QQ[\re(\z), \im(\z)]$;
        \item all right-hand sides of atoms in $\varphi$ are $0$.
    \end{enumerate}
    Then $\NNNRR_\z$ is a system of normal forms for $\FFF_\z$, and all normal forms in $\NNNRR_\z$ are real formulas.
\end{proposition}

\begin{proof}
    Let $\varphi \in \FFF_\z$. We transform $\varphi$ in three subsequent steps as follows:
    \begin{enumerate}
        \item Equivalently transform all left-hand side and right-hand side terms of atoms in $\varphi$ into Cartesian normal form according to Proposition~\ref{PR:cartnf}.
        Note that within inequalities $r_1 \leq r_2$ and $r_1 < r_2$, where $r_1$, $r_2$ are real terms, we obtain normal forms in $\QQ[\re(\z), \im(\z)]$.
        This yields $\varphi'$ with $\CC \models \varphi' \mylongleftrightarrow \varphi$.
        \item Equivalently rewrite all atoms in $\varphi'$ as follows, preserving Cartesian normal form:
        \begin{equation*}
            f_1 + i \cdot g_1 \mathrel{\rho} f_2 + \ii \cdot g_2 \to (f_1 - f_2) + \ii \cdot (g_1 - g_2) \mathrel{\rho} 0,\quad \rho \in \{=, \neq, \leq, <\}.
        \end{equation*}
        This yields $\varphi''$ with $\CC \models \varphi'' \mylongleftrightarrow \varphi$. All right-hand sides of atoms in $\varphi''$ are $0$, and all left-hand sides of inequalities in $\varphi''$ are still polynomials in $\QQ[\re(\z), \im(\z)]$.
        \item Equivalently rewrite equations and disequalities in $\varphi''$ as follows:
        \begin{equation*}
            f + \ii \cdot g = 0 \to {f = 0} \mathrel{\myland} {g = 0},\quad
            f + \ii \cdot g \neq 0 \to {f \neq 0} \mylor {g \neq 0}.
        \end{equation*}
        This yields $\varphi'''$ with $\CC \models \varphi''' \mylongleftrightarrow \varphi$, and $\varphi'''$ has the properties (i) and (ii) specified in the lemma.\qedhere
    \end{enumerate}
\end{proof}

\section{Quantifier Elimination, Completeness,\\ and Decidability}\label{SE:qe}

\begin{theorem}[Pseudo-complex quantifier elimination]\label{TH:qe}
Let $\varphi \in \FFF_\z$ and assume without loss of generality that $\varphi$ is in real normal form. Apply the following three steps to compute from $\varphi$ a quantifier free formula $\varphi' \in \QQQ_\z$.
\begin{enumerate}
    \item \textbf{Purification.} Replace all subterms $\re(z_i)$ and $\im(z_i)$ in $\varphi$ with auxiliary variables $z^\re_i$ and $z^\im_i$, respectively. Furthermore, replace all quantifiers $\exists z_i$ and $\forall z_i$ in $\varphi$ with $\exists z^\re_i\exists z^\im_i$ and $\forall z^\re_i\forall z^\im_i$, respectively. This yields a first-order formula $\psi$ in the language of ordered rings.
    \item \textbf{Real quantifier elimination.} Apply \emph{any} real \qe procedure to obtain a quantifier-free formula $\psi'$ in the language of ordered rings such that $\RR \models \psi \mylongleftrightarrow \psi'$.
    \item \textbf{Substitution.} Remove the auxiliary variables via substitution of the original terms. This yields $\varphi' \in \QQQ_\z$ where $\varphi' = \psi'[z^\re_i \gets \re(z_i), z^\im_i \gets \im(z_i)]$.
\end{enumerate}
Then $\CC \models \varphi \mylongleftrightarrow \varphi'$. In other words, our three steps establish a quantifier elimination procedure for the complex numbers in our language $\LLL$.
\end{theorem}

\begin{proof}
    Consider extended formulas $\varphi(\z)$, $\varphi'(\z)$ and $\psi(\z^\re, \z^\im)$, $\psi'(\z^\re, \z^\im)$, and their interpretations
    \begin{equation*}
        \varphi(\z)^*, \varphi'(\z)^*: \CC^n \to \{0, 1\},\quad
        \psi(\z^\re, \z^\im)^*, \psi'(\z^\re, \z^\im)^*: \RR^{2n} \to \{0, 1\}.
    \end{equation*}

    As an auxiliary lemma, we prove by strong induction on the number $q$ of prenex quantifiers in $\varphi$ that
    \begin{equation}\label{EQ:th1claim}
        \varphi(\z)^*(\cc) = \psi(\z^\re, \z^\im)^*(\re^*(\cc), \im^*(\cc))\quad
        \text{for all}\quad
        \cc \in \CC^n.
    \end{equation}
    For $q=0$, the claim \eqref{EQ:th1claim} is immediate. Assume $q > 0$ and let $\cc \in \CC^n$. Assume without loss of generality that $\varphi = (\exists z \tilde\varphi)$ and $\psi = (\exists z^\re \exists z^\im \tilde\psi)$ with extended formulas $\tilde\varphi(z, \z)$ and $\tilde\psi(z^\re, \z^\re, z^\im, \z^\im)$ and corresponding interpretations on $\CC^{n+1}$ and $\RR^{2n+2}$, respectively. By the induction hypothesis we have
    \begin{equation}\label{EQ:Th1ih}
        \tilde\varphi(z, \z)^*(c, \cc)
        = \tilde\psi(z^\re, \z^\re, z^\im, \z^\im)^*(\re^*(c), \re^*(\cc), \im^*(c), \im^*(\cc)).
    \end{equation}
    Combining \eqref{EQ:Th1ih} with the formal semantics of quantifiers (see the enumeration on p.~\pageref{PAGE:semantics}), we obtain
    \begin{align}\label{EQ:minmax}
        (\exists z \tilde\varphi)(\z)^*(\cc)
        & = \max_{c \in \CC}\bigl(\tilde\varphi(z,\z)^*(c,\cc)\bigr)\notag\\
        & = \max_{c\in \CC}\bigl(
            \tilde\psi(
                z^\re, \z^\re, z^\im, \z^\im)^*(\re^*(c), \re^*(\cc), \im^*(c), \im^*(\cc))\bigr)\notag\\
        & = \max_{r_1, r_2 \in \RR}\bigl(
            \tilde\psi(
                z^\re, \z^\re, z^\im, \z^\im)^*(r_1, \re^*(\cc), r_2, \im^*(\cc))\bigr)\notag\\
        &= (\exists z^\re \exists z^\im \tilde\psi)(\z^\re, \z^\im)^*(\re^*(\cc), \im^*(\cc)).
    \end{align}
    With universal quantifiers $\varphi = (\forall z \tilde\varphi)$ and $\psi = (\forall z^\re \exists z^\im \tilde\psi)$ the same argument holds with $\min$ instead of $\max$ in \eqref{EQ:minmax}. This proves our auxiliary lemma \eqref{EQ:th1claim}. Hence,
    \begin{align*}
        \varphi(\z)^*(\cc)
        & = \psi(\z^\re, \z^\im)^*(\re^*(\cc), \im^*(\cc)) & & \text{purification, using \eqref{EQ:th1claim}}\\
        & = \psi'(\z^\re, \z^\im)^*(\re^*(\cc), \im^*(\cc)) & & \text{real \qe}\\
        & = \varphi'(\z)^*(\cc) & & \text{substitution}
    \end{align*}
    for all $\cc \in \CC^n$.
\end{proof}

\begin{corollary}[Completeness and decidability]
    The theory $\Theta = \{\,\vartheta \in \FFF_{()} \mid \CC \models \vartheta\,\}$ is complete and decidable. Our pseudo-complex quantifier elimination procedure yields a decision procedure for $\Theta$.
\end{corollary}

\begin{proof}
    When applying the \qe procedure in the proof of Theorem~\ref{TH:qe} to a sentence $\vartheta \in \FFF_{()}$, the purified formula $\psi$ is a sentence, too. Modulo obvious simplifications, real \qe applied to $\psi$ will yield either $\top$ or $\bot$, which is preserved by the subsequent substitution step. In the former case, $\CC \models \vartheta \mylongleftrightarrow \top$ and thus $\vartheta \in \Theta$; else $\CC \models \vartheta \mylongleftrightarrow \bot$, equivalently $\CC \models \lnot\vartheta \mylongleftrightarrow \top$, and thus $\vartheta \notin \Theta$ and $\lnot\vartheta \in \Theta$.
\end{proof}

\section{Complex Reinterpretation}\label{SE:reinter}

We present a method for simplifying the output formula of our quantifier elimination procedure by combining real atoms of the form $r_1 = 0 \myland r_2 = 0$ into a single atom $r_1 + \ii \cdot r_2 = 0$. Subsequent conversion of all terms into conjugate normal form can result in a more compact formula. In the special case $\re(z) = 0 \myland \im(z) = 0$, the resulting terms even collapse to a single variable, yielding $z = 0$.

For our presentation here we assume without loss of generality that the input formula of our method is a conjunction of real atoms, and we focus on equations. It is easy to see that a dual variant with disjunctions and disequalities exists.\footnote{In practice, we employ a simplification framework with implicit theories similar to \cite{DolzmannSturm:97b}, which is beyond the scope of this article. Simpler implementations could rely on Boolean normal form computations. Our approach smoothly integrates the conjunctive and disjunctive variants while avoiding exponential blow-up.} Note that ordering inequalities cannot be straightforwardly combined, as they are equivalent to $\bot$ in the presence of non-real terms. Moreover, disequalities cannot be directly combined within a conjunction.

We are thus given a set $R = \{r_1 = 0, \ldots, r_n = 0\}$ of equations, where each $r_i$ is a real term, and our goal is to find a pairing of elements of $R$ that minimizes the total size of the resulting merged equations together with the remaining unpaired equations. This can be viewed as a variant of the minimum edge cover problem in the following sense.

\begin{definition}[Minimum cost partial edge cover]\label{DE:mcpec}
Let $G = (V, E)$ be an undirected graph, $c_V \colon V \to \NN$ a cost function on the vertices, and $c_E \colon E \to \NN$ a cost function on the edges. For any subset $S \subseteq E$, we denote by $V(S) \subseteq V$ the set of all endpoints of edges in $S$. A \emph{minimum cost partial edge cover (\mcpec)} is a subset $S \subseteq E$ that minimizes $c(S) = \sum_{v \in V \setminus V(S)} c_V(v) + \sum_{e \in S} c_E(e)$.
\end{definition}

We fix a system of normal forms $\NNN_\z$ for atomic formulas $\AAA_\z$ and denote by $|\alpha|$ the word length of $\alpha \in \AAA_\z$. We choose $V = R$ with a cost function $c_V(r_i = 0) = |\NNN_\z(r_i = 0)|$, and $E = \{\,\{v, w\} \mid v, w \in V, v \neq w\,\}$ with a cost function
\begin{equation}\label{EQ:minimize}
    c_E(\{r_i = 0, r_j = 0\})
    = \min \{ |\NNN(r_i + \ii \cdot r_j = 0)|, |\NNN(r_j + \ii \cdot r_i = 0)| \}.
\end{equation}
For any edge $e \in E$, we denote by $\mu(e)$ the equation that minimizes \eqref{EQ:minimize}. The following proposition shows how to use a solution $S \subseteq E$ of the \mcpec problem to construct a formula that is equivalent to $\bigwedge R$.

\begin{proposition}
Let $G = (V, E)$ be the graph constructed from a set of equations $R$ and let $S \subseteq E$.
Then $\CC \models \bigwedge R \mylongleftrightarrow \bigwedge \bigl(V \setminus V(S)\bigr) \myland \bigwedge_{e \in S} \mu(e)$.
\end{proposition}
\begin{proof}
If an equation $r_i = 0$ is covered by multiple edges in $S$, we use idempotence to replace it with $r_i = 0 \myland \dots \myland r_i = 0$ on the left-hand side of the equivalence. Then the statement follows directly since $\CC \models r_i = 0 \myland r_j = 0 \mylongleftrightarrow \mu(e)$ for all $\{r_i = 0, r_j = 0\} \in S$ by construction.
\end{proof}

\begin{example}
    Fix a system of normal forms $\NNN_\z$ of $\AAA_\z$ such that all atoms are of the form $t \mathrel{\rho} 0$ where either $t = 0$ or $t \in \QQ[\ii][\z,\cj{\z}]$ with leading coefficient $1$.
Consider the formula
\begin{equation}\label{EQ:exmcpecin}
    \re(x) = 0 \myland \im(x) = 0 \myland \re(y) = 0 \myland \im(y) > 0.
\end{equation}
{\samepage
We have $(V, E)$ with costs $c_V$, $c_E$ as follows:
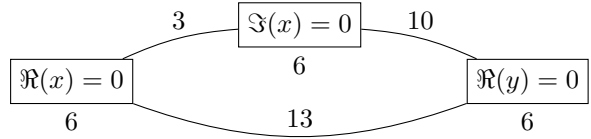
\begin{center}
\smallskip
\begin{tikzpicture}[scale=1, transform shape]
\tikzset{vertex/.style={rectangle, draw, minimum size=1.5em}}
\node[vertex, label=below:{$6$}] (A) at (0,0)  {$\re(x) = 0$};
\node[vertex, label=below:{$6$}] (B) at (3,0.75) {$\im(x) = 0$};
\node[vertex, label=below:{$6$}] (C) at (6,0) {$\re(y) = 0$};
\draw (A) to[bend left=10] node[midway, above] {$3$}  (B);
\draw (B) to[bend left=10] node[midway, above] {$10$} (C);
\draw (C) to[bend left=20] node[midway, above] {$13$} (A);
\end{tikzpicture}
\end{center}}
For instance, $c_V(\re(x) = 0) = |x + \cj{x} = 0| = 6$ and
\begin{equation*}
    c_E(\{\re(x) = 0, \im(x) = 0\})  = \min \{ |x = 0|, |\cj{x} = 0| \} = 3.
\end{equation*}
The \mcpec is given by $S = \bigl\{\{\re(x) = 0, \im(x) = 0\}\bigr\}$ corresponding to the merged equation $x = 0$ with cost $3$, and the remaining unpaired equation $\re(y) = 0$ with cost $6$. This yields the following equivalent of \eqref{EQ:exmcpecin}:
\begin{equation*}
    x = 0 \myland y + \cj{y} = 0 \myland \ii \cj{y} - \ii y > 0.
\end{equation*}
\end{example}

The \mcpec problem can be efficiently solved using ideas from combinatorial optimization, specifically reduction to the maximum weight matching problem \cite{Schrijver:02}. We denote by $E(v)$ the set of all edges incident to a vertex $v \in V$. A subset $M \subseteq E$ is a \emph{matching} if no two edges in $M$ share a common endpoint. A \emph{maximum weight matching} is a matching $M$ that maximizes $c_E(M) = \sum_{e \in M} c_E(e)$.

\begin{lemma}[Complexity of the \mcpec problem with low vertex costs]
Consider an instance of the \mcpec problem for a graph $G = (V, E)$ with cost functions $c_V$ and $c_E$. Assume that for all vertices $v \in V$ and edges $e \in E(v)$, the cost functions satisfy $c_V(v) \leq c_E(e)$. Then an \mcpec can be computed in time $O(n \cdot (m + n \log n))$, where $n = |V|$ and $m = |E|$ are the numbers of vertices and edges in $G$, respectively.
\end{lemma}
\begin{proof}
We construct a new graph $G' = (V', E')$ with cost function $c' \colon E' \to \ZZ$ as follows. Set $V' = V$, $E' = E$, and for each edge $e \in E$ with endpoints $v$ and $w$, define $c'(e) = c_V(v) + c_V(w) - c_E(e)$. We claim that every maximum weight matching $M \subseteq E$ in $G'$ is an \mcpec in $G$ and can thus be computed in time $O\bigl(n \cdot (m + n \log n)\bigr)$ \cite{Gabow:18}.

To prove the claim, we first consider a partial edge cover $S \subseteq E$ in $G$ and show that it can be transformed into a matching without increasing its partial edge cover cost in $G$. Assume $S$ is not a matching. Then there exist two edges $e_1$, $e_2 \in S$ that share a common endpoint $v \in V$. Denote by $w$ the other endpoint of $e_1$. Since $c_V(w) \leq c_E(e_1)$ by assumption, we can remove $e_1$ from $S$ while not increasing the total cost of $S$. Repeating this process until no two edges in $S$ share a common endpoint results in a matching.

We come back to our initial claim that every maximum weight matching $M$ in $G'$ is an \mcpec in $G$. Assume for a contradiction that there exists a partial edge cover $S$ in $G$ with $c(S) < c(M)$. By the first part of our proof, we can assume without loss of generality that $S$ is a matching. Therefore, we have
\begin{align*}
    c(S)
    &= \sum_{v \in V \setminus V(S)} c_V(v) + \sum_{e \in S} c_E(e)
    = \sum_{v \in V} c_V(v) - \sum_{v \in V(S)} c_V(v) + \sum_{e \in S} c_E(e) \\
    &= \sum_{v \in V} c_V(v) - \sum_{e \in S} c'(e)
    = \sum_{v \in V} c_V(v) - c'(S)
\end{align*}
and, similarly, $c(M) = \sum_{v \in V} c_V(v) - c'(M)$. Since $c(S) < c(M)$, this implies $c'(S) > c'(M)$, which contradicts the assumption that $M$ is a maximum weight matching in $G'$.
\end{proof}

We finally show that complexity of the \mcpec problem in the general case can be reduced to the special case considered in the previous lemma.

\begin{proposition}[Complexity of the \mcpec problem]
The \mcpec problem for a graph $G = (V, E)$ with cost functions $c_V$ and $c_E$ can be solved in time $O\bigl(n \cdot (m + n \log n)\bigr)$, where $n = |V|$ and $m = |E|$ are the number of vertices and edges in $G$.
\end{proposition}
\begin{proof}
We construct a new graph $G' = (V', E')$ with $V' = V$, $E' = E$, $c_E' = c_E$, and a new vertex cost function $c_V' \colon V' \to \NN$, $c_V'(v) = \min \{ c_V(v), c_E(e_v) \}$, where $e_v \in E$ is an edge with minimum cost among all edges with endpoint $v$. By construction, $G'$ satisfies the assumptions of the previous lemma, and an \mcpec $S'$ in $G'$ can be computed in time $O\bigl(n \cdot (m + n \log n)\bigr)$. Our final \mcpec in $G$ is then given $S = S' \cup S''$ with $S'' = \{ e_v \mid v \in V \setminus V(S'), \, c_E(e_v) < c_V(v) \}$. It remains to show that $S$ is indeed an \mcpec in $G$. First, note that $S'$ and $S''$ are disjoint, so that
\begin{equation*}
    c(S) = \sum_{v \in V \setminus V(S)} c_V(v) + \sum_{e \in S'} c_E(e)
    + \sum_{e \in S''} c_E(e).
\end{equation*}
By definition of $S''$ and $c_V'$, we have $S'' \subseteq V \setminus V(S') \cap V(S'')$ and
\begin{equation*}
    \sum_{e \in S''} c_E(e) = \sum_{e_v \in S''} c_V'(v) \leq \sum_{v \in V \setminus V(S') \cap V(S'')} c'_V(v).
\end{equation*}
Moreover, $V \setminus V(S) = V \setminus V(S') \cap V \setminus V(S'')$
and $c_V(v) = c_V'(v)$ for all $v \in V \setminus V(S)$ by the definition of $S''$ and $c_V'$. Therefore,
\begin{equation*}
    \sum_{v \in V \setminus V(S)} c_V(v) = \sum_{v \in V \setminus V(S') \cap V \setminus V(S'')} c'_V(v).
\end{equation*}
By combining the previous equation and inequality, we obtain
\begin{equation*}
    \sum_{v \in V \setminus V(S)} c_V(v) + \sum_{e \in S''} c_E(e)
    \leq \sum_{v \in V \setminus V(S')} c'_V(v).
\end{equation*}
Using $c'_E = c_E$, we conclude that $c(S) \leq c'(S')$. Since $S'$ is an \mcpec in $G'$, and because $c'(T) \leq c(T)$ for any $T \subseteq E$ by definition of $c'$, it follows that $c(S) \leq c'(S') \leq c'(T) \leq c(T)$ for any partial edge cover $T$ in $G$. Hence, $S$ is an \mcpec in $G$.
\end{proof}

\section{Complexity}\label{SE:complexity}
We summarize the asymptotic complexity of the quantifier elimination problem for the theory $\ACF$ of algebraically closed fields of characteristic zero and the theory $\RCF$ of real closed fields. Furthermore, we examine the complexity of available algorithms, with a particular emphasis on well-supported implementations. On this basis, we address some implications for the framework proposed here.

With respect to word length, the asymptotic worst-case size of the \qe output in both $\ACF$ and $\RCF$ is doubly exponential in the size of the input, and this bound is tight. This yields doubly exponential time complexity \cite{Heintz:83a,DavenportHeintz:88a,Weispfenning:88a}. It is noteworthy that elimination procedures for $\ACF$ admit exponential-space bounds \cite{ChistovGrigoriev1984combined}, and the corresponding \qe algorithms can be implemented within exponential space provided the output is generated incrementally. When the number of alternations between existential and universal quantifiers in the input is bounded, the time complexity for both $\ACF$ and $\RCF$ is singly exponential \cite{10.1007/BFb0030287,DBLP:journals/jsc/Renegar92combined}. The complexity depends primarily on the number of quantifier alternations and the number of quantified variables, which determine the exponent in the standard bounds for \qe in $\ACF$ \cite{10.1007/BFb0030287} and $\RCF$ \cite{Grigoriev:88a,DBLP:journals/jsc/Renegar92combined,BasuPollack:96a}. Asymptotically, these bounds have the same form for $\ACF$ and $\RCF$.

\begin{remark}[Asymptotic complexity of pseudo-complex \qe]
    Asymptotic complexity bounds for real \qe carry over to pseudo-complex \qe. During the purification step in Theorem~\ref{TH:qe}, the input size for real \qe may grow polynomially in the original input size. More specifically, the number of quantifiers in the input doubles, while the number of quantifier alternations remains unchanged. Subsequent complex reinterpretation of the \qe result is polynomial in the size of the result and thus dominated by the complexity bounds for \qe.\qed
\end{remark}

From a practical perspective, \qe in $\ACF$ can be approached via comprehensive Gröbner bases (\textsc{cgb}) \cite{Weispfenning:92a}. The Redlog system \cite{DolzmannSturm:97a} implements this by computing disjunctive normal forms at each quantifier alternation and employing \textsc{cgb} computations with disjoint case distinctions on parameter vanishing, which leads to extremely rapid growth of the resulting quantifier-free formulas. Compared to $\ACF$, \qe in $\RCF$ is supported more extensively across a wide range of software environments, including systems such as Qepcad, Redlog, Maple, and Mathematica. Numerous applications in the sciences and engineering are documented in the literature; see e.g.~\cite{DolzmannSturm:99a,Sturm:17a} and the references therein. For general-purpose applications, the method of choice remains partial cylindrical algebraic decomposition \cite{CollinsHong:91}. This algorithm has doubly exponential complexity in the total number of variables occurring in the input, regardless of the presence of quantifier alternations or even whether a variable is quantified \cite{Brown:2007:CQE:1277548.1277557}.

On the one hand, the strong available software support for real \qe facilitates the integration of our approach into the corresponding software environments. On the other hand, the lack of efficient implementations of complex \qe suggests that our method may be of interest even for complex \qe problems without our additional operations.

\section{Implementation and Examples}\label{SE:implementation}
We have implemented our approach in our Python-based open-source system Logic1.\footnote{\url{https://github.com/logic1-eu/logic1}} For the real \qe part, Logic1 uses by default an implementation of virtual substitution following essentially \cite{Kosta:16a}. Note, that real \qe serves as a black box so that it is not hard to plug in alternative implementations in the future. An interface to Redlog \cite{DolzmannSturm:97a} already exists. Our implementation allows to dynamically switch between conjugate and Cartesian normal forms.

The computation times for all subsequent examples are summarized in Table~\ref{tab:examples-runtime} on p.\pageref{tab:examples-runtime}, which also indicates which of the two available normal forms for terms was chosen in each case.

\begin{example}[Cartesian coordinates]\label{ex:cartesian-coordinates}
The representation of a complex number $z$ in Cartesian coordinates can be expressed by the following formula:
\begin{equation*}
    \varphi =  \forall z \exists x \exists y \bigl(\im(x) = 0 \myland \im(y) = 0 \myland z = x + \ii \cdot y\bigr).
\end{equation*}
Quantifier elimination yields the quantifier-free formula $\top$.
\end{example}

\begin{example}[Roots of unity]\label{ex:roots-of-unity}
In contrast to the real numbers, the complex numbers contain a square root of $-1$. This can be expressed as $\varphi_1 = \exists z(z^2 + 1 = 0)$, for which quantifier elimination yields the quantifier-free formula $\top$. More generally,
Weispfenning discusses the equivalence of the following formula to $d^4 - 1 = 0$ as an example for complex \qe using comprehensive Gröbner bases \cite{Weispfenning:92a}:
\begin{equation*}
    \varphi_2 = \exists c \forall b \forall a \bigl((a = d \myland b = c) \mylor (a = c \myland b = 1) \mylongrightarrow b = a^2 \bigr).
\end{equation*}
Our \qe applied to $\varphi_2$ yields $d + 1 = 0 \mylor d + i = 0 \mylor d - 1 = 0 \mylor d - i = 0$, corresponding to the factorization of $d^4 - 1$ over $\CC$.
\end{example}

\begin{example}[Counterexample for geometry provers]\label{ex:counterexample-geometry-provers}
In his famous monograph \cite{Chou:88a}, Chou presents the following example of a formula that holds over $\RR$ but not over $\CC$:
\begin{equation*}
    \varphi = \forall x_{1} \forall x_{2} \bigl(x_{1}^2 + x_{2}^2 = 1 \myland x_{1} = 2 \mylongrightarrow x_{2} = 1\bigr).
\end{equation*}
Indeed, the unit circle $x_1^2 + x_2^2 = 1$ does not intersect with the line $x_1 = 2$, hence the left side of the implication is equivalent to $\bot$. Our \qe yields the quantifier-free formula $\bot$ over $\CC$.
\end{example}

Hilbert spaces generalize the Euclidean space $\RR^n$ to a more abstract setting where there still exists a notion of angles, captured via an inner product. They are fundamental in science and engineering, particularly in the formulation of quantum mechanics \cite{Shankar:94} as well as in various applications in signal processing \cite{KennedySadeghi:13} and machine learning \cite{SchoelkopfSmola:01}. We consider the Hilbert space $\CC^n$ with an inner product $\langle \cdot, \cdot \rangle: \CC^n \times \CC^n \to \CC$ defined by
\begin{equation*}
    \langle \vv, \ww \rangle = \sum_{i=1}^n v_i \cdot \cj{w_i}
\end{equation*}
for all $\vv = (v_1, \dots, v_n) \in \CC^n$ and $\ww = (w_1, \dots, w_n) \in \CC^n$.

\begin{example}[Orthogonality]\label{ex:orthogonality}
Two vectors $\vv, \ww \in \CC^n$ are called \emph{orthogonal} if $\langle \vv, \ww \rangle = 0$. Consider the following formula with free variables $v_1$, \dots, $v_n$:
\begin{equation*}
    \varphi = \forall w_1 \dots \forall w_n \bigl(\langle \vv, \ww \rangle = 0\bigr).
\end{equation*}
We obtain the quantifier-free formula $v_1 = 0 \myland v_2 = 0 \myland v_3 = 0$ for $n = 3$. This proves a special case of the more general fact that $\vv \in \CC^n$ is zero if and only if it is orthogonal to all vectors in $\CC^n$. Table~\ref{tab:orthogonality-series-runtime} summarizes computation times for growing values of $n$.
\end{example}

\begin{table}
\centering
\begin{tabularx}{0.9\textwidth}{c*{8}{>{\centering\arraybackslash}X}}
\hline
$n$ & 10 & 15 & 20 & 25 & 30 & 35 & 40 & 45\\
\hline
time (s) & 1.58 & 3.78 & 7.42 & 12.89 & 21.06 & 32.49 & 48.59 & 69.28 \\
\hline
\end{tabularx}
\caption{Wall clock computation times of Example~\ref{ex:orthogonality} (Orthogonality) for increasing dimension $n$. All computations have been performed on an 8 + 4 core M4 MacBook Pro, Nov 2024, with 24~GB RAM\label{tab:orthogonality-series-runtime}}
\end{table}

\begin{example}[Cauchy--Schwarz inequality]\label{ex:cauchy-schwarz-inequality}
The Cauchy--Schwarz inequality is fundamental in linear algebra and functional analysis. It can be expressed by the following formula with free variables $v_1$, \dots,~$v_n$:
\begin{equation*}
    \varphi = \forall w_1 \dots \forall w_n \bigl(\langle \vv, \ww \rangle \cdot \cj{\langle \vv, \ww \rangle} \leq \langle \vv, \vv \rangle \cdot \langle \ww, \ww \rangle\bigr),
\end{equation*}
where $\langle \vv, \ww \rangle \cdot \cj{\langle \vv, \ww \rangle} = |\langle \vv, \ww \rangle|^2$. We obtain the quantifier-free formula $\top$ for $n$ from $1$ to $4$.
\end{example}

In quantum mechanics, the state of a physical system is represented by a vector in a complex Hilbert space $H$, and physical observables correspond to self-adjoint linear operators $A \colon H \to H$. The expected value of the measurement outcome of an observable $A$ on a state $\vv \in H$ is  $\langle A \vv, \vv \rangle$. In the special case of a single qubit, the state space is given by $H = \CC^2$.

\begin{example}[Self-adjoint matrices]\label{ex:self-adjoint-matrices}
Consider a matrix $A = (a_{ij})_{i,j=1}^n \in \CC^{n \times n}$. Then $A$ is self-adjoint if and only if
\begin{equation*}
    \varphi_1 = \forall v_1 \dots \forall v_n \forall w_1 \dots \forall w_n \bigl(\langle A \vv, \ww \rangle = \langle \vv, A \ww \rangle\bigr).
\end{equation*}
Applying \qe in the case of a single qubit, i.e. $n = 2$, we obtain the equivalent quantifier-free formula $a_{11} - \cj{a_{11}} = 0 \myland a_{12} - \cj{a_{21}} = 0 \myland a_{22} - \cj{a_{22}} = 0$. This is a special case of the well-known characterization of self-adjoint matrices as those that are equal to their conjugate transpose, i.e. $a_{ij} = \cj{a_{ji}}$ for all $i$, $j \in \{1, \dots, n\}$. Alternatively, self-adjointness can be characterized by the following formula:
\begin{equation*}
    \varphi_2 = \forall v_1 \dots \forall v_n \bigl(\im(\langle A \vv, \vv \rangle) = 0\bigr).
\end{equation*}
Applying \qe again in the case $n = 2$, we obtain the same quantifier-free equivalent as above. In particular, this shows that the expected measurement outcome of a quantum observable $A$ on a single qubit is always real, i.e. physically meaningful.
\end{example}

\begin{example}[Density matrices]\label{ex:density-matrices}
A probabilistic mixture of states of a quantum system can be represented by a density matrix, which is a self-adjoined, positive semidefinite matrix with trace $1$. In the case of a single qubit, a self-adjoined matrix $A \in \CC^{2 \times 2}$ with trace $1$ has the form
\begin{equation*}
    A = \begin{pmatrix}
    \re(a) & b\\
    \cj{b} & 1 - \re(a)
    \end{pmatrix}
    \quad \text{for $a$, $b \in \CC$}.
\end{equation*}
The positive semidefiniteness of $A$ can be expressed by the following formula:
\begin{equation*}
    \varphi = \forall v_1 \forall v_2 \bigl(\langle A \vv, \vv \rangle \geq 0\bigr).
\end{equation*}
Note that $\langle A \vv, \vv \rangle \in \RR$ for all $\vv \in \CC^2$, which can be proved using a normal form computation. Applying \qe, we obtain the following quantifier-free formula:
\begin{multline*}
            \re(a) \geq 0
     \myland   \bigl(\re(a) = 0 \mylor \re(a)^{2} + \re(\mathit{b})^{2} + \im(b)^{2} - \re(a) \leq 0\bigr) \\
     {}\myland \bigl(\re(a) > 0 \mylor \re(b) = 0\bigr)
     \myland \bigl(\re(a) > 0 \mylor \re(b) \leq 0\bigr)
     \myland \bigl(\re(a) > 0 \mylor \re(b) \geq 0\bigr)\\
     {}\myland \bigl(\re(a) > 0 \mylor \im(b) = 0 \mylor \re(b)^{2} \im(b) + \im(b)^{3} \leq 0\bigr)\\
     {}\myland \bigl(\re(a) > 0 \mylor \im(b) = 0 \mylor \re(b)^{2} \im(b) + \im(b)^{3} \geq 0\bigr)\\
     {}\myland \bigl(\re(a) > 0 \mylor \im(b) \leq 0\bigr)
     \myland\bigl(\re(a) > 0 \mylor \im(b) \geq 0\bigr).
\end{multline*}
This can be further simplified by hand to $\re(b)^{2} + \im(b)^{2} \leq \re(a) \cdot (1 - \re(a))$.
\end{example}

Another important application of complex numbers in science and engineering is in the analysis of electrical circuits \cite{Nilsson:14}. A special class of such circuits are $RC$ filters, which are designed to allow signals of certain frequencies to pass while damping signals of other frequencies. In the analysis of filters, their transfer function $H(s) = {p(s)}/{q(s)}$ plays a central role, where $p(s)$ and $q(s)$ are polynomials in a complex variable $s$; see \cite{Oppenheim:97}.

The \qe in our following examples using external assumptions \cite{DSW:98}. The input is a first-order formula $\varphi$ plus a list $A$ of atoms. The output is a quantifier-free formula $\varphi'$ such that $\CC \models \bigwedge A \mylongrightarrow (\varphi \mylongleftrightarrow \varphi')$. Note that this yields $\CC \models \varphi \mylongleftrightarrow \varphi'$ as usual when $A$ is empty.

\begin{example}[Gain of passive $RC$ high-pass filter]\label{ex:gain-passive-rc-high-pass-filter}
Figure~\ref{fig:passive-rc-high-pass-filter} shows a circuit diagram of a passive $RC$ high-pass filter, which consists of a capacitor $C$ and a resistor $R$. It has the transfer function $H(s) = {p(s)}/{q(s)}$ with $p(s) = R Cs$ and $q(s) = RCs + 1$. The gain at a frequency $\omega \in \RR$ is given by $|H(\omega\ii)|$. The following formula expresses that the gain of the filter is bounded by $g$ for all frequencies:
\begin{equation*}
    \varphi = \forall s \bigl(\Re(s) = 0 \mylongrightarrow p(s)\cj{p(s)} < g^2 \cdot q(s)\cj{q(s)}\bigr),
\end{equation*}
where $p(s)\cj{p(s)} = |p(s)|^2$ and $q(s)\cj{q(s)} = |q(s)|^2$. Under the assumptions
\begin{equation*}
    A = \{\re(R) > 0,\, \im(R) = 0,\, \re(C) > 0,\, \im(C) = 0,\, \re(g) > 0,\, \im(g) = 0\},
\end{equation*}
we obtain the quantifier-free formula
\begin{equation*}
    \re(g)^2 - 1 \geq 0.
\end{equation*}
Recall that there is an assumption $\im(g) = 0$ in $A$.

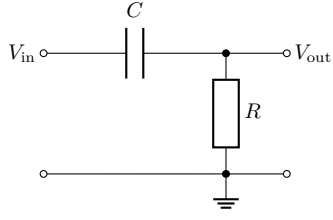
\begin{figure}
\centering
\begin{circuitikz}[european, scale=0.8, transform shape]
\draw
  (0,0) node[left]{$V_{\text{in}}$}
  to[short, o-] (1,0)
  to[C, l=$C$] (2,0)
  to[short, -o] (4,0)
  node[right]{$V_{\text{out}}$};
\draw
  (3, 0)
  node[circ]{} to[R, l=$R$] (3,-2)
  node[circ]{} node[ground]{};
\draw
    (0,-2) to[short, o-o] (4,-2);
\end{circuitikz}
\caption{Circuit diagram of a passive $RC$ high-pass filter}
\label{fig:passive-rc-high-pass-filter}
\end{figure}
\end{example}

\begin{example}[Stability of active $RC$ filter]\label{ex:stability-active-rc-filter}
Figure~\ref{fig:active-rc-filter} shows a circuit diagram of the active $RC$ filter from \cite[Example 1]{Gielen:94}. It consists of resistors $R_1$, \dots,~$R_{11}$, capacitors $C_1$, $C_2$ and four operational amplifiers. Denote by $G_i = {1}/{R_i}$ the conductance of resistor $R_i$. Then the transfer function is given by $H(s) = {p(s)}/{q(s)}$ with $p(s)$ and $q(s)$ as follows:
\begin{align*}
    p(s) & = -G_4G_8(G_1G_2G_9 + G_1G_3G_9 + G_1G_9G_1 + G_2G_6G_9) \\
    & \phantom{=} \quad + G_7C_2(G_1G_3G_9 + G_1G_3G_1) \cdot s  - G_2G_7C_1C_2G_9 \cdot s^2, \\
    q(s) &= G_1G_9G_4G_6G_8 + G_1G_9G_5G_7G_2 \cdot s + G_1G_9G_7G_1G_2 \cdot s^2.
\end{align*}
A circuit is called \emph{stable} if its transfer function has no poles in the right half-plane. The following formula gives a sufficient condition with free variables $G_1$, \dots,~$G_9$, $C_1$, $C_2$:
\begin{equation*}
    \varphi = \forall s \bigl(q(s) = 0 \mylongrightarrow \Re(s) < 0\bigr).
\end{equation*}
Quantifier elimination yields the quantifier-free formula $\top$ under the physically motivated assumption that all free variables stand for positive real numbers, i.e.,
\begin{equation*}
    A = \bigl\{\re(G_i) > 0,\, \im(G_i) = 0\bigr\}_{i \in \{1, \dots, 9\}}
    \cup \bigl\{\re(C_j) > 0,\, \im(C_j) = 0\bigr\}_{j \in \{1, 2\}}.
\end{equation*}

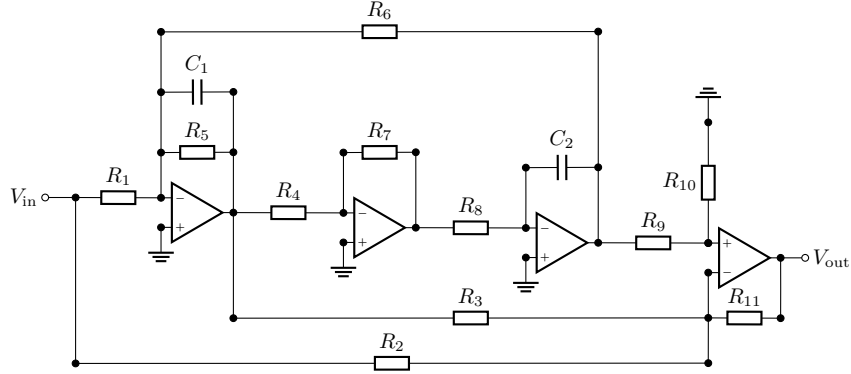
\begin{figure}
\centering
\begin{circuitikz}[european, scale=0.8, transform shape]
\ctikzset{resistors/scale=0.5}
\ctikzset{capacitors/scale=0.5}
\draw
    (3, -0.25) node[op amp, scale=0.5] (opamp1) {}
    (6, -0.5) node[op amp, scale=0.5] (opamp2) {}
    (9, -0.75) node[op amp, scale=0.5] (opamp3) {}
    (12, -1) node[op amp, scale=0.5, noinv input up] (opamp4) {}
    (opamp1.+) node[circ] {} node[ground] {}
    (opamp2.+) node[circ] {} node[ground] {}
    (opamp3.+) node[circ] {} node[ground] {};
\draw
    (0.5,0) node[left]{$V_{\text{in}}$} to[short, o-]
    (1, 0)     to[R, l=$R_1$, european, *-*] (opamp1.-)
    (opamp1.out) to[R, l=$R_4$, *-*] (opamp2.-)
    (opamp2.out) to[R, l=$R_8$, *-*] (opamp3.-)
    (opamp3.out) to[R, l=$R_9$, *-*] (opamp4.+)
    (opamp4.out) to[short, *-o] (13, -1) node[right]{$V_{\text{out}}$};
\draw
    ($(opamp1.-) + (0,0.75)$) to[R, l=$R_5$, *-*] ($(opamp1.out) + (0,1)$)
    ($(opamp1.-) + (0,1.75)$) to[C, l=$C_1$, *-*] ($(opamp1.out) + (0,2)$) to[short] (opamp1.out)
    (opamp1.-) to[short] ($(opamp1.-) + (0,2.75)$) to[R, l=$R_6$, *-*] ($(opamp3.out) + (0,3.5)$) to[short] (opamp3.out)
    (opamp2.-) to[short] ($(opamp2.-) + (0,1)$) to[R, l=$R_7$, *-*] ($(opamp2.out) + (0,1.25)$) to[short] (opamp2.out)
    (opamp3.-) to[short] ($(opamp3.-) + (0,1)$) to[C, l=$C_2$, *-*] ($(opamp3.out) + (0,1.25)$)
    (opamp4.+) to[R, l=$R_{10}$, *-*] ($(opamp4.+) + (0,2)$) node[ground, rotate=180]{}
    (opamp1.out) to[short] ($(opamp1.out) - (0,1.75)$) to[R, l=$R_3$, *-*] ($(opamp4.-) - (0,0.75)$)
    (1, 0) to[short] (1, -2.75) to[R, l=$R_2$, *-*] ($(opamp4.-) - (0,1.5)$)
    to[short, -*] (opamp4.-)
    ($(opamp4.-) - (0,0.75)$) to[R, l=$R_{11}$, -*] ($(opamp4.out) - (0,1)$)
    to[short, -] (opamp4.out)
    ;
\end{circuitikz}
\caption{Circuit diagram of an active $RC$ filter; see \cite[Example 1]{Gielen:94}}
\label{fig:active-rc-filter}
\end{figure}
\end{example}

\begin{table}
    \centering
    \begin{tabular}{r@{\quad}l@{\quad}c@{\quad}r}
    \hline
    &  & normal form used & time (s) \\
    \hline
    \ref{ex:cartesian-coordinates} & Cartesian coordinates & conjugate & 0.04 \\
    \ref{ex:roots-of-unity} & Roots of unity, $\varphi_1$ & conjugate & 0.03 \\
    \ref{ex:roots-of-unity} & Roots of unity, $\varphi_2$ & conjugate & 1.23 \\
    \ref{ex:counterexample-geometry-provers} & Counterexample for geometry provers & conjugate & 0.07 \\
    \ref{ex:orthogonality} & Orthogonality ($n = 3$) & conjugate & 0.2 \\
    \ref{ex:cauchy-schwarz-inequality} & Cauchy-Schwarz inequality ($n = 4$) & conjugate & 24.57 \\
    \ref{ex:self-adjoint-matrices} & Self-adjoint matrices, $\varphi_1$ & conjugate & 1.28 \\
    \ref{ex:self-adjoint-matrices} & Self-adjoint matrices, $\varphi_2$ & conjugate & 0.91 \\
    \ref{ex:density-matrices} & Density matrices & Cartesian & 1.79 \\
    \ref{ex:gain-passive-rc-high-pass-filter} & Gain of passive $RC$ high-pass filter & Cartesian & 2.64 \\
    \ref{ex:stability-active-rc-filter} & Stability of active $RC$ filter & conjugate & 3.76\\
    \hline
    \end{tabular}
    \caption{Summary of the wall-clock computation times of our Examples~\ref{ex:cartesian-coordinates}--\ref{ex:stability-active-rc-filter}. All computations have been performed on an 8 + 4 core M4 MacBook Pro, Nov 2024, with 24~GB RAM. We also annotate the term normal form chosen for the respective example\label{tab:examples-runtime}}
\end{table}

\section{Conclusions}\label{SE:conclusions}
We have introduced a formal logical framework for the complex numbers that includes symbols for the imaginary unit, as well as for real and imaginary parts and complex conjugation. Within this framework, additional operations, such as the complex absolute value, become definable. The corresponding first-order theory is complete and decidable, and it admits effective quantifier elimination. From a technical perspective, our approach builds on quantifier elimination over the real numbers, allowing existing implementations to be used as a black box. Compared with the classical theory of algebraically closed fields of characteristic zero in the language of rings, our framework is often better aligned with the needs of applications in the sciences and engineering. To illustrate its practical utility, we have presented several academic examples computed using an implementation in our Python-based open-source system Logic1.

\section*{Acknowledgments}
Nicolas Faroß is currently supported by SSF (Swedish Foundation for Strategic Research), grant number FUS21-0063.

\bibliographystyle{plainnat}
\bibliography{pseudo-complex-qe}

\begin{thebibliography}{39}
\providecommand{\natexlab}[1]{#1}
\providecommand{\url}[1]{\texttt{#1}}
\expandafter\ifx\csname urlstyle\endcsname\relax
  \providecommand{\doi}[1]{doi: #1}\else
  \providecommand{\doi}{doi: \begingroup \urlstyle{rm}\Url}\fi

\bibitem[Aubry et~al.(1999)Aubry, Lazard, and
  Moreno~Maza]{AubryLazardMorenoMaza1999}
Philippe Aubry, Daniel Lazard, and Marc Moreno~Maza.
\newblock On the theories of triangular sets.
\newblock \emph{Journal of Symbolic Computation}, 28\penalty0 (1--2):\penalty0
  105--124, 1999.
\newblock \doi{10.1006/jsco.1999.0269}.

\bibitem[Basu et~al.(1996)Basu, Pollack, and Roy]{BasuPollack:96a}
Saugata Basu, Richard Pollack, and Marie-Fran\c{c}oise Roy.
\newblock On the combinatorial and algebraic complexity of quantifier
  elimination.
\newblock \emph{Journal of the ACM}, 43\penalty0 (6):\penalty0 1002--1045,
  October 1996.
\newblock \doi{10.1145/235809.235813}.

\bibitem[Brown and Davenport(2007)]{Brown:2007:CQE:1277548.1277557}
Christopher~W. Brown and James~H. Davenport.
\newblock The complexity of quantifier elimination and cylindrical algebraic
  decomposition.
\newblock In \emph{Proceedings of the 2007 International Symposium on Symbolic
  and Algebraic Computation (ISSAC '07)}, pages 54--60. ACM, 2007.
\newblock \doi{10.1145/1277548.1277557}.

\bibitem[Chen and Moreno~Maza(2012)]{ChenMorenoMaza2012}
Changbo Chen and Marc Moreno~Maza.
\newblock Algorithms for computing triangular decompositions of polynomial
  systems.
\newblock \emph{Journal of Symbolic Computation}, 47\penalty0 (6):\penalty0
  610--642, 2012.
\newblock \doi{10.1016/j.jsc.2011.12.023}.

\bibitem[Chevalley(1943)]{Chevalley1943}
Claude Chevalley.
\newblock On the theory of local rings.
\newblock \emph{Annals of Mathematics}, 44\penalty0 (4):\penalty0 690--708,
  1943.
\newblock \doi{10.2307/1969105}.

\bibitem[Chistov and Grigor'ev(1984)]{10.1007/BFb0030287}
A.~L. Chistov and D.~Yu. Grigor'ev.
\newblock Complexity of quantifier elimination in the theory of algebraically
  closed fields.
\newblock In M.~P. Chytil and V.~Koubek, editors, \emph{Mathematical
  Foundations of Computer Science 1984}, volume 176 of \emph{Lecture Notes in
  Computer Science}, pages 17--31. Springer, 1984.
\newblock ISBN 978-3-540-38929-3.
\newblock \doi{10.1007/BFb0030287}.

\bibitem[Chistov and Grigor'ev(1991)]{ChistovGrigoriev1984combined}
A.~L. Chistov and D.~Yu. Grigor'ev.
\newblock Subexponential-time solving systems of algebraic equations.
\newblock \emph{Leningrad Mathematical Journal}, 2\penalty0 (6), 1991.
\newblock Parts I and II. English translation of Zap. Nauchn. Sem. LOMI 137
  (1984).

\bibitem[Chou(1988)]{Chou:88a}
Shang-Ching Chou.
\newblock \emph{Mechanical Geometry Theorem Proving}, volume~41 of
  \emph{Mathematics and Its Applications}.
\newblock D. Reidel Publishing Company, Dordrecht, 1988.
\newblock ISBN 978-90-277-2650-6.

\bibitem[Collins(1975)]{Collins:75}
George~E. Collins.
\newblock Quantifier elimination for real closed fields by cylindrical
  algebraic decomposition.
\newblock In H.~Brakhage, editor, \emph{Automata Theory and Formal Languages.
  2nd GI Conference}, volume~33 of \emph{Lecture Notes in Computer Science},
  pages 134--183. Springer, 1975.
\newblock \doi{10.1007/3-540-07407-4_17}.

\bibitem[Collins and Hong(1991)]{CollinsHong:91}
George~E. Collins and Hoon Hong.
\newblock Partial cylindrical algebraic decomposition for quantifier
  elimination.
\newblock \emph{Journal of Symbolic Computation}, 12\penalty0 (3):\penalty0
  299--328, September 1991.
\newblock \doi{10.1016/S0747-7171(08)80152-6}.

\bibitem[Davenport and Heintz(1988)]{DavenportHeintz:88a}
James~H. Davenport and Joos Heintz.
\newblock Real quantifier elimination is doubly exponential.
\newblock \emph{Journal of Symbolic Computation}, 5\penalty0 (1--2):\penalty0
  29--35, February--April 1988.
\newblock \doi{10.1016/S0747-7171(88)80004-X}.

\bibitem[Davis(1954)]{Davis:54a}
Martin Davis.
\newblock Final report on mathematical procedures for decision problems.
\newblock Technical report, Institute for Advanced Study, Princeton, NJ,
  October 1954.

\bibitem[Dolzmann and Sturm(1997{\natexlab{a}})]{DolzmannSturm:97a}
Andreas Dolzmann and Thomas Sturm.
\newblock Redlog: Computer algebra meets computer logic.
\newblock \emph{ACM SIGSAM Bulletin}, 31\penalty0 (2):\penalty0 2--9, June
  1997{\natexlab{a}}.
\newblock \doi{10.1145/261320.261324}.

\bibitem[Dolzmann and Sturm(1997{\natexlab{b}})]{DolzmannSturm:97b}
Andreas Dolzmann and Thomas Sturm.
\newblock Simplification of quantifier-free formulae over ordered fields.
\newblock \emph{Journal of Symbolic Computation}, 24\penalty0 (2):\penalty0
  209--232, 1997{\natexlab{b}}.
\newblock \doi{https://doi.org/10.1006/jsco.1997.0123}.

\bibitem[Dolzmann et~al.(1998)Dolzmann, Sturm, and Weispfenning]{DSW:98}
Andreas Dolzmann, Thomas Sturm, and Volker Weispfenning.
\newblock A new approach for automatic theorem proving in real geometry.
\newblock \emph{Journal of Automated Reasoning}, 21\penalty0 (3):\penalty0
  357--380, December 1998.
\newblock \doi{10.1023/A:1006031329384}.

\bibitem[Dolzmann et~al.(1999)Dolzmann, Sturm, and
  Weispfenning]{DolzmannSturm:99a}
Andreas Dolzmann, Thomas Sturm, and Volker Weispfenning.
\newblock Real quantifier elimination in practice.
\newblock In B.~H. Matzat, G.-M. Greuel, and G.~Hiss, editors,
  \emph{Algorithmic Algebra and Number Theory}, pages 221--247. Springer, 1999.
\newblock \doi{10.1007/978-3-642-59932-3_11}.

\bibitem[Gabow(2018)]{Gabow:18}
Harold~N. Gabow.
\newblock Data structures for weighted matching and extensions to b-matching
  and f-factors.
\newblock \emph{ACM Transactions on Algorithms}, 14\penalty0 (3):\penalty0
  1--80, 2018.
\newblock \doi{10.1145/3183369}.

\bibitem[Gielen et~al.(1994)Gielen, Wambacq, and Sansen]{Gielen:94}
G.~Gielen, P.~Wambacq, and W.M. Sansen.
\newblock Symbolic analysis methods and applications for analog circuits: a
  tutorial overview.
\newblock \emph{Proceedings of the IEEE}, 82\penalty0 (2):\penalty0 287--304,
  1994.
\newblock \doi{10.1109/5.265355}.

\bibitem[Grigoriev(1988)]{Grigoriev:88a}
D.~Yu. Grigoriev.
\newblock Complexity of deciding {T}arski algebra.
\newblock \emph{Journal of Symbolic Computation}, 5\penalty0 (1--2):\penalty0
  65--108, February--April 1988.
\newblock \doi{10.1016/S0747-7171(88)80006-3}.

\bibitem[Heintz(1983)]{Heintz:83a}
Joos Heintz.
\newblock Definability and fast quantifier elimination in algebraically closed
  fields.
\newblock \emph{Theoretical Computer Science}, 24\penalty0 (3):\penalty0
  239--277, August 1983.
\newblock \doi{10.1016/0304-3975(83)90002-6}.

\bibitem[Kennedy and Sadeghi(2013)]{KennedySadeghi:13}
Rodney~A. Kennedy and Parastoo Sadeghi.
\newblock \emph{Hilbert Space Methods in Signal Processing}.
\newblock Cambridge University Press, 2013.
\newblock \doi{10.1017/CBO9780511844515}.

\bibitem[Ko{\v s}ta(2016)]{Kosta:16a}
Marek Ko{\v s}ta.
\newblock \emph{New Concepts for Real Quantifier Elimination by Virtual
  Substitution}.
\newblock Doctoral dissertation, Saarland University, Germany, December 2016.

\bibitem[Le and Safey El~Din(2021)]{LeSafeyElDin2021}
Huu~Phuoc Le and Mohab Safey El~Din.
\newblock Faster one block quantifier elimination for regular polynomial
  systems of equations.
\newblock In \emph{Proceedings of the 2021 International Symposium on Symbolic
  and Algebraic Computation (ISSAC '21)}, pages 265--272. ACM, 2021.
\newblock \doi{10.1145/3452143.3465546}.

\bibitem[Nilsson and Riedel(2014)]{Nilsson:14}
James~W. Nilsson and Susan Riedel.
\newblock \emph{Electric Circuits}.
\newblock Pearson Education, 2014.
\newblock ISBN 9780133594812.

\bibitem[Oppenheim et~al.(1997)Oppenheim, Willsky, and Nawab]{Oppenheim:97}
Alan~V. Oppenheim, Alan~S. Willsky, and S.~Hamid Nawab.
\newblock \emph{Signals \& Systems}.
\newblock Prentice-Hall signal processing series. Prentice Hall, 1997.
\newblock ISBN 9780138147570.

\bibitem[Rahkooy and Sturm(2021)]{RahkooySturm:21a}
Hamid Rahkooy and Thomas Sturm.
\newblock Parametric toricity of steady state varieties of reaction networks.
\newblock In \emph{Computer Algebra in Scientific Computing (CASC 2021)},
  volume 12865 of \emph{Lecture Notes in Computer Science}, pages 314--333.
  Springer, 2021.
\newblock \doi{10.1007/978-3-030-85165-1_18}.

\bibitem[Renegar(1992)]{DBLP:journals/jsc/Renegar92combined}
James Renegar.
\newblock On the computational complexity and geometry of the first-order
  theory of the reals.
\newblock \emph{Journal of Symbolic Computation}, 13\penalty0 (3):\penalty0
  255--352, March 1992.
\newblock \doi{10.1016/S0747-7171(10)80003-3}.
\newblock Parts I--III.

\bibitem[Schrijver(2002)]{Schrijver:02}
Alexander Schrijver.
\newblock \emph{Combinatorial Optimization}, volume~24 of \emph{Algorithms and
  Combinatorics}.
\newblock Springer Berlin, Heidelberg, 2002.
\newblock ISBN 9783540443896.

\bibitem[Schölkopf and Smola(2001)]{SchoelkopfSmola:01}
Bernhard Schölkopf and Alexander~J. Smola.
\newblock \emph{Learning with Kernels: Support Vector Machines, Regularization,
  Optimization, and Beyond}.
\newblock The MIT Press, 2001.
\newblock \doi{10.7551/mitpress/4175.001.0001}.

\bibitem[Seidenberg(1954)]{Seidenberg1954}
Abraham Seidenberg.
\newblock A new decision method for elementary algebra.
\newblock \emph{Annals of Mathematics}, 60\penalty0 (2):\penalty0 365--374,
  1954.
\newblock \doi{10.2307/1969640}.

\bibitem[Shankar(1994)]{Shankar:94}
R.~Shankar.
\newblock \emph{Principles of Quantum Mechanics}.
\newblock Springer New York, 1994.
\newblock \doi{10.1007/978-1-4757-0576-8}.

\bibitem[Sturm(2017)]{Sturm:17a}
Thomas Sturm.
\newblock A survey of some methods for real quantifier elimination, decision,
  and satisfiability and their applications.
\newblock \emph{Mathematics in Computer Science}, 11\penalty0 (3--4):\penalty0
  483--502, December 2017.
\newblock \doi{10.1007/s11786-017-0319-z}.

\bibitem[Tarski(1930)]{Tarski1930}
Alfred Tarski.
\newblock {\"U}ber einige fundamentale {B}egriffe der {M}etamathematik.
\newblock \emph{Comptes Rendus des S{\'e}ances de la Soci{\'e} des Sciences et
  des Lettres de Varsovie, Classe III}, 23:\penalty0 22--29, 1930.

\bibitem[Tarski(1954)]{Tarski1954}
Alfred Tarski.
\newblock Contributions to the theory of models. {I}.
\newblock \emph{Indagationes Mathematicae}, 57:\penalty0 572--581, 1954.
\newblock \doi{10.1016/S1385-7258(54)50074-0}.

\bibitem[Tarski(1957)]{Tarski:48a}
Alfred Tarski.
\newblock A decision method for elementary algebra and geometry. {P}repared for
  publication by {J}.~{C}.~{C}.~{M}c{K}insey.
\newblock RAND Report R109, August 1, 1948, Revised May 1951, Second Edition,
  RAND, Santa Monica, CA, 1957.

\bibitem[Weispfenning(1988)]{Weispfenning:88a}
Volker Weispfenning.
\newblock The complexity of linear problems in fields.
\newblock \emph{Journal of Symbolic Computation}, 5\penalty0 (1--2):\penalty0
  3--27, February--April 1988.
\newblock \doi{10.1016/S0747-7171(88)80003-8}.

\bibitem[Weispfenning(1992)]{Weispfenning:92a}
Volker Weispfenning.
\newblock Comprehensive {G}r{\"o}bner bases.
\newblock \emph{Journal of Symbolic Computation}, 14\penalty0 (1):\penalty0
  1--29, July 1992.
\newblock \doi{10.1016/0747-7171(92)90023-W}.

\bibitem[Weispfenning(1997)]{Weispfenning:97b}
Volker Weispfenning.
\newblock Quantifier elimination for real algebra---the quadratic case and
  beyond.
\newblock \emph{Applicable Algebra in Engineering, Communication and
  Computing}, 8\penalty0 (2):\penalty0 85--101, January 1997.
\newblock \doi{10.1007/s002000050055}.

\bibitem[Weispfenning(1998)]{Weispfenning:98a}
Volker Weispfenning.
\newblock A new approach to quantifier elimination for real algebra.
\newblock In B.~F. Caviness and J.~R. Johnson, editors, \emph{Quantifier
  Elimination and Cylindrical Algebraic Decomposition}, Texts and Monographs in
  Symbolic Computation, pages 376--392. Springer, 1998.
\newblock \doi{10.1007/978-3-7091-9459-1_20}.

\end{thebibliography}

\end{document}